\documentclass[11pt]{article}
\usepackage{graphics}
\usepackage{latexsym}
\usepackage{color}
\usepackage{epsfig}

\usepackage{amsthm}
\usepackage{graphicx}
\usepackage{color}
\usepackage{amsmath}
\usepackage{amssymb}
\usepackage{amscd}
\usepackage{latexsym}
\usepackage{graphics}
\usepackage{subfigure}
\usepackage{url}
\usepackage{afterpage,longtable,lscape}
\usepackage{verbatim}  

\theoremstyle{plain}
\newtheorem{theorem}{Theorem}[section]
\newtheorem{lemma}[theorem]{Lemma}
\newtheorem{corollary}[theorem]{Corollary}

\DeclareMathOperator{\rev}{rev}

\begin{document}

\title{In-place permuting and perfect shuffling using involutions}

\author{
   QINGXUAN YANG\footnotemark[1]
   \and JOHN ELLIS\footnotemark[2] \and KHALEGH MAMAKANI\footnotemark[2] \and FRANK RUSKEY\footnotemark[2]
   }

\maketitle

\renewcommand{\thefootnote}{\fnsymbol{footnote}}
\footnotetext[1]{Department of Computer Science, Peking University, China.}
\footnotetext[2]{Department of Computer Science, University of Victoria, Canada.}
\renewcommand{\thefootnote}{\arabic{footnote}}

\begin{abstract}
Every permutation of $\{1,2,\ldots,N\}$ can be written as the product of
  two involutions.
As a consequence, any permutation of the elements of an array can be performed in-place
  using simultaneous swaps in two rounds of swaps.
In the case where the permutation is the $k$-way perfect shuffle we develop two methods for
  efficiently computing the pair of involutions that accomplishes these swaps.

The first method works whenever $N$ is a power of $k$; in this case the time is
  $O(N)$ and space $O(\log^2 N)$.
The second method applies to the general case where $N$ is a multiple of $k$;
  here the time is $O(N \log N)$ and the space is $O(\log^2 N)$.
If $k=2$ the space usage of the first method can be reduced to $O(\log N)$ on a machine that
  has a SADD (population count) instruction.
\end{abstract}

\noindent
\textbf{keywords: }{Perfect shuffle, merging, sorting, permutation, involution.}


\section{Introduction}

Certain useful permutations on strings, represented by linear arrays of $N$ elements, can be realised,
without difficulty, in $O(N)$ time and using only $O(\log N)$ extra space to store a constant number of variables.
Examples are reversals and cyclic shifts.
However, the permutation known as the perfect shuffle is not so easily achieved within the given time and space
  constraints.

The perfect shuffle is defined on strings of even length, say $N$.
Assuming the string elements occupy positions 0 through $N-1$, the element occupying
position $i$ is translated to position $2i \mod N-1$.
The effect is illustrated in Figure \ref{shuffleExample}.
This version of the perfect shuffle is called an \emph{in-shuffle}; the end elements
are not moved.
The version in which they are moved is called an \emph{out-shuffle}.
It is only necessary to analyse one form and in this paper we only consider the in-shuffle.
In the generalised, $k$-way, shuffle, the element occupying position $i$ is translated
to position $ki \mod N-1$, where $N = kM$ for some integer $M$.

\begin{figure}[h]
\label{shuffleExample}
\begin{center}
\textbf{a b c d e f 1 2 3 4 5 6 $\rightarrow$ a 1 b 2 c 3 d 4 e 5 f 6} \\
\end{center}
\caption{A Perfect In-shuffle on 12 elements.}
\end{figure}

The perfect shuffle arises in several computer science applications.
For example, certain parallel processing algorithms are conveniently executed
on a perfect-shuffle-based architecture \cite{IB-A945016} and
there is an \emph{in-place} merging algorithm \cite{Ellis:2000:ISMPS,DAT:2011} which reduces
the list merging problem to the problem of realising the perfect shuffle (see also \cite{ESA}).

Shuffle algorithms that conform to the $O(N)$ time and $O(\log N)$ space constraints have been
described in \cite[Section 7]{Ellis:2000:ISMPS} and in \cite{PJ:2008}.
These solutions use a ``cycle leader'' algorithm which requires the generation of
``leader'' for each of the cycles that make up the permutation.
An analysis of the cycles comprising the perfect shuffle is given in \cite{Ellis:2000:CCPSP} and
in \cite{EFS:2002} the $k$-way perfect shuffle on a deck of size $N = kM$ is discussed.

Our methods are not cycle leader algorithms; rather they are based on identifying pairs of independent elements that swap
positions.
These swaps occur in two rounds, and thus each element participates in no more that two swaps.
The resulting configurations can be represented by diagrams that are reminiscent of comparator networks used
  for sorting, but now the comparators always swap their elements.
One might imagine that this would be a good way to permute data in hardware since the configurations can be
  laid out on two layers, with the ``wires" on one layer going horizontally, and going vertically in the other
  layer.  See Figure \ref{fig:examples}.

We first present a general strategy for permuting in-place in two rounds of swaps (Section 2)
  and then present two instantiations
  of that strategy (Sections 3 and 4) as applied to the perfect shuffle permutation.
The two methods are sufficiently simple and efficient that they may provide a viable
  alternative to known methods.
The first method works whenever $N$ is a power of $k$; in this case the time is
  $O(N)$ and space $O(\log^2 N)$.
If $k=2$ we also show how to adapt it, still in $O(N)$ time, for general $N$ (but not in 2 rounds of swaps).
The second method applies to the general case where $N$ is only assumed to be a multiple of $k$;
here the time is $O(N \log N)$ and the space is $O(\log^2 N)$.
The space usage of the first method can be reduced to $O(\log N)$ on a machine that
has a SADD (population count) instruction.
The time complexity measures just stated assume that swap, assignment, relational operations and
  elementary arithmetic operations all take $O(1)$ time.
For space complexity we count the number of bits required to represent a variable.
Furthermore, all the swaps in each phase are independent
  and so, in parallel, the perfect shuffle can be achieved in $O(1)$ time.

\section{General permutations}

An \emph{involution} is a permutation that is its own inverse.
It is well-known that every permutation can be written as the product of disjoint cycles.
Every involution is the disjoint product of 1-cycles and 2-cycles.
If the permutation consists of one cycle, then we call it a \emph{circular permutation}.

We will prove in this section that every permutation is the product of two involutions.
Thus every permutation can be performed in two rounds of (simultaneous) swaps, or sequentially by having
each element participate in at most two swaps.
Permuting in place has been considered previously \cite{Fich:1995:PP}.
It has been noted before that every permutation is the product of two sets of transpositions, i.e. cycles of length 2.
See, for example, exercise 10.1.17 in \cite{Scott:1964}, although we do not know the originator of this result.
In the theorem below we also count the number of factorizations of a circular permutation as a product
of involutions.
This result was also found recently and independently in \cite{PetersenTenner}, see their Corollary 2.5.

\begin{theorem}
Every circular permutation of $n$ elements can be written as the product of two involutions.
Furthermore, there are exactly $n$ such factorizations if $n > 2$.
\label{thm:circinv}
\end{theorem}
\begin{proof}
Since we can relabel the elements of a cycle, it is sufficient to prove the result for $C_n$,
where $C_n$ is the cyclic shift of $0,1, \ldots, n-1$ written in cycle notation as
$C_n = (0\ 1\ \cdots \ n{-}2\ n{-}1)$.
Define  $I_k$ to be the involution
\[
I_k := (0\ k)(1\ k-1) \cdots (\lfloor k/2 \rfloor\ \lceil k/2 \rceil)\
(k{+}1\ n{-}1)(k{+}2\ n{-}2) \cdots ( \lfloor (k+n)/2 \rfloor\ \lceil (k+n)/2 \rceil) \]
It is easy to check that for, $k = 0,1,\ldots,n-1$, the the following equality holds
(where $k-1$ is interpreted as $n-1$ when $k = 0$).
\begin{equation}
C_n = I_{k} I_{k-1}.
\end{equation}
Furthermore, we show below that there are no other factorizations of $C_n$ into two involutions.

Suppose that $C_n = ST$ where $S$ and $T$ are involutions, which we will think of
as products of disjoint 2-cycles, where a 2-cycle can be degenerate in the sense of
being of the form $(x\ x)$.
The involution $S$ must have a 2-cycle $(0\ k)$ for some $k$.
But then $T$ must have the 2-cycle $(0\ k{-}1)$ in order that $ST(k{-}1) = S(0) = k$.
And then $S$ must have the 2-cycle $(1\ k{-}1)$ in order that $ST(0) = S(k{-}1) = 1$.
Continuing in this manner, alternately inferring 2-cycles in $S$ and $T$ we conclude
that $S$ contains $(0\ k)(1\ k-1) \cdots (\lfloor k/2 \rfloor\ \lceil k/2 \rceil)$
and that $T$ contains $(0\ k{-}1)(1\ k-2) \cdots (\lfloor (k{-}1)/2 \rfloor\ \lceil (k{-}1)/2 \rceil)$.

We now claim that $T$ must have the 2-cycle $(k\ n{-}1)$ in order that $ST(n{-}1) = S(k) = 0$.
But then $S$ must have the 2-cycle $(k{+}1\ n{-}1)$ in order that $ST(k) = S(n{-}1) = k{+}1$.
Continuing in this manner, alternately inferring 2-cycles in $T$ and $S$ we conclude eventually
that $T$ contains $(k\ n{-}1)(k{+}1\ n{-}2) \cdots ( \lfloor (k+n-1)/2 \rfloor\ \lceil (k+n-1)/2 \rceil)$
and $S$ contains $(k{+}1\ n{-}1)(k{+}2\ n{-}2) \cdots ( \lfloor (k+n)/2 \rfloor\ \lceil (k+n)/2 \rceil)$.

Thus we have shown that $S = I_k$ and $T = I_{k-1}$.
\end{proof}


\begin{corollary}
Every permutation of $n$ elements can be written as the product of two involutions.
\label{cor:involution}
\end{corollary}

\begin{proof}
This follows from Theorem \ref{thm:circinv} since it is well-known that every permutation
can be written as the product of disjoint cycles, and the product of
disjoint involutions is again an involution.  In more detail, suppose that
$\pi = C_1 C_2 \cdots C_k$ is the disjoint cycle decomposition of a permutation $\pi$.
By the theorem, $\pi = (S_1 T_1)(S_2 T_2) \cdots (S_k T_k)$ where $S_i$ and $T_i$ are
involutions.  Since they are disjoint, $T_i S_j = S_j T_i$ whenever $i \neq j$.
Thus $\pi = (S_1 S_2 \cdots S_k)(T_1 T_2 \cdots T_k)$, the product of two involutions.
\end{proof}

In order to apply this theorem and its corollary the cycles of the permutation need to be identified.
In the next sections we will show how to compute the two involutions by two separate methods when the permutation
is the perfect shuffle.

\section{Computing perfect shuffles}

Initially we assume that $N = k^n$, where we are doing a $k$-way perfect shuffle.
Let $i = (b_{n-1} \cdots b_1 b_0)_k$ be an $n$-digit $k$-ary number.
Define $\rev_p(i)$ to be the integer whose base-$k$ representation is the same as that of $i$
except that the $p$ least significant digits are reversed.
For example, in binary, $\rev_3(44) = \rev_3((101100)_2) = (101001)_2 = 41$.
Since $\rev_p( \rev_p(i) ) = i$, the function $\rev_p(i)$ is an involution.

Now consider the map $\Xi(i) := \rev_n( \rev_{n-1}(i))$ operating on the $n$-digit $k$-ary number $i$.
Let $j = b_{n-1}$.  Note that
\begin{eqnarray*}
\Xi(i)
   & = & \rev_n( \rev_{n-1}( (j b_{n-2} \cdots b_1 b_0)_k )) \\
   & = & \rev_n( (j b_0 b_1 \cdots b_{n-2})_k ) \\
   & = & (b_{n-2} \cdots b_1 b_0 j)_k \\
   & = & ki - j(k^n -1).
\end{eqnarray*}
Hence,
\[
\Xi(i) \cong ki \bmod{(k^n-1)},
\]
and thus $\Xi(i)$ is the $k$-way ``perfect shuffle" permutation of $\{0,1,\ldots,k^n-1\}$.
The perfect shuffle can then be seen as two sets of swaps, where the first performs
$\rev_{n-1}$ and the second performs $\rev_{n}$.
Since all of the swaps in each of the two phases are independent,
  in parallel, the perfect shuffle can be achieved in $O(1)$ time.
Example networks are shown in Figure \ref{fig:examples}.
There the digits in the right side column represent $i$, and those in the left side
represent the inverse of the function.

\begin{figure}[t]
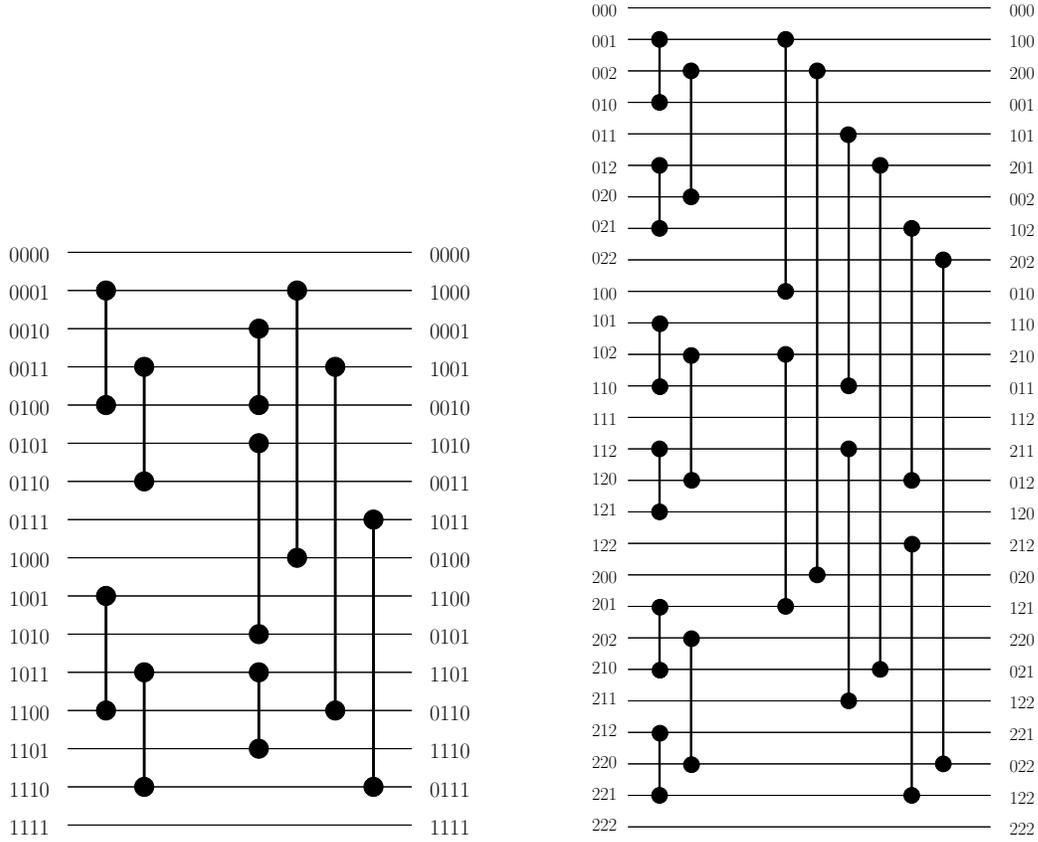

\begin{center}
\scalebox{0.40}{\input{YE-2-4.pstex_t}}
\hspace{0.5in}
\scalebox{0.33}{\input{YE-3-3.pstex_t}}
\end{center}
\caption{A network corresponding to the algorithm for (a)
  $k = 2$, $n=4$ on the left, and (b) $k=3$, $n=3$ on the right.}
\label{fig:examples}
\end{figure}

As a sequential algorithm one can execute \textsc{revSwap}$(n,k,n-1)$; \textsc{revSwap}$(n,k,n)$
where \textsc{revSwap}$(n,k,t)$ is given below.
(The $a :=: b$ is the notation introduced by Knuth and used in Sedgewick \cite{Sedgewick} to
denote the swap or exchange operation, i.e. the values of $a$ and $b$ are exchanged.)

The pseudo-code below outlines our algorithm. \\

\begin{tabular}{lcr}
\begin{minipage}[t]{2in}
\begin{tabbing}
xxx \= xxx \= xxx \kill
\textsc{Shuffle}($n,k$) \\
\> \textsc{revSwap}$(n,k,n-1)$ \\
\> \textsc{revSwap}$(n,k,n)$ \\
\end{tabbing}
\end{minipage}
& \hspace{0.3in} &
\begin{minipage}[t]{3in}
\begin{tabbing}
\textsc{revSwap}$(n,k,t)$ \\
xxx \= xxx \= xxx \kill
\> \textbf{for} $i \in \{0,1,\ldots,k^n-1\}$ \textbf{do} \\
\> \> $j := \rev_t(i)$ \\
\> \> \textbf{if} $i<j$ \textbf{then} $A[i] :=: A[j]$ \\
\end{tabbing}
\end{minipage}
\end{tabular}

To generate the $(i,\rev_n(i))$ pairs,
we use the following lemma to adjust $\rev_n(i))$ appropriately as $i$
is incremented.

\begin{lemma}
Let $p$ be the leftmost digit that changes in the $k$-ary representation of $i$ when $i$ is incremented by one.
Then
\[
\rev_n(i+1) = \rev_n(i) - k^n + k^{n-p} + k^{n-p-1}.
\]
\end{lemma}
\begin{proof}
As a $k$-ary number, for some $0 \le j < k-1$, the numbers $i$ and $i+1$ have the form
\[
i = (\cdots j \underbrace{(k{-}1) (k{-}1) \cdots (k{-}1)}_{p} )_k
\ \ \text{ and }\ \  i{+}1 = (\cdots (j{+}1) \underbrace{00 \cdots 0}_{p} )_k .
\]
Thus
\begin{equation*}
\begin{split}
\rev_n(i) - \rev_n(i{+}1) & = \left( j k^{n-p-1} + (k{-}1)(k^{n-p}+\cdots+k^{n-1}) \right) - \left( (j+1) k^{n-p-1} \right) \\
& =  - k^{n-p-1} + (k{-}1) \frac{ k^n - k^{n-p} }{k{-}1} \\
& =  k^n - k^{n-p-1} - k^{n-p}.
\end{split}
\end{equation*}
\end{proof}


We can realise the computation defined by this Lemma, ignoring the computation of $p$ for the moment,
by storing a table of the $n$ powers of $k$, using $O(\log^2 N)$ bits, plus $O(1)$ bits for a constant number
of variables, and in constant time per pair.

The successive $p$ values can be computed by simulating the incrementation of $i$
on a $k$-ary counter and noting which is the leftmost bit to change.
This process is constant time per incrementation, when amortised over all values of $i$.
As a function of $i$ when $k=2$, the value $p$ is called the ``ruler function", see  \cite{Knu:2011:ACP}, and can
be computed in various other ways.
On a machine with a ``population count" or SADD instruction (e.g., the MMIX machine of Knuth \cite{Knu:2011:ACP}),
it can be computed in a constant number of machine operations per incrementation.
It can also be implemented by using a cast on floating point numbers to extract $\lfloor \lg n \rfloor$,
  as pointed out in Knuth \cite{Knu:2011:ACP} on page 142 (see equation (55)).  
We implemented this approach in \texttt{C} on a
machine with a $2.0$ GHz Opteron processor and found that it computed the required swaps in
under a minute when $N = 2^{30} = 1,073,741,824$.  The program is available at
\url{http://webhome.cs.uvic.ca/~ruskey/Publications/Shuffle/PerfectShuffle.html}.


How many swaps are done by the algorithm?
If $n$ is even then the number of swaps in the first phase is $k(k^{n-1}-k^{n/2})/2$
and the number in the second phase is $(k^{n}-k^{n/2})/2$.
If $n$ is odd then the number of swaps in the first phase is $k(k^{n-1}-k^{(n-1)/2})/2$
and the number in the second phase is also $(k^{n}-k^{(n+1)/2})/2$.
In the even case the total is $k^n-(k+1)k^{n/2}/2$ and in the odd case
the total is $k^n - k^{(n+1)/2}$.  Thus in either case the number of swaps is
  $N + O(\sqrt{N})$.

\subsection{Modification for the cases where $N$ is not a power of $k$}

The preceding analysis solves the problem for the case where the number of elements, $N$, is a power of $k$.
The general case, where $N$ is not a power of $k$, but still we have $N = kM$, can be reduced to the special case using the technique
described in \cite[Section 7]{Ellis:2000:ISMPS}, at least when $k = 2$.

The technique may be understood by means of an example.
Suppose that $N = 2M = 30 = 2\cdot 15$.  Think of the two initial sequences of
15 elements as being broken down into subsequences of elements whose sizes are
powers of 2 listed in decreasing order according to the binary representation of $M$.
Proceeding from most significant bit to least significant bit we can do a
series of rotations until the subsequences are in their proper positions
as they would appear in the perfect shuffle.  This process is illustrated below.

\medskip
\noindent
Initial string (spaces are only for clarity): \\
\texttt{ssssssss tttt uu v ssssssss tttt uu v} \\
Move 8 \texttt{s}s: \\
\texttt{ssssssss \underline{ssssssss tttt uu v} tttt uu v} \\
Move 4 \texttt{t}s: \\
\texttt{ssssssss ssssssss tttt \underline{tttt uu v} uu v} \\
Move 2 \texttt{u}s: \\
\texttt{ssssssss ssssssss tttt tttt uu \underline{uu v} v} \\

We would now apply our previous perfect shuffle algorithm to the \texttt{s}s, the
 \texttt{t}s and the \texttt{u}s.
These smaller perfect shuffles could be done in parallel.

In the general case, there is no need to do a rotation if the corresponding bit in the
binary representation of $M$ is 0.
The elements that were moved in each of the four rotations in the example above are underlined.
Assuming that $M = (b_{s} \cdots b_1 b_0)_2$, the number of element changes, $C$, that occur in the successive moves is
\[
C = \sum_{i=1}^{s} b_i (b_{i} \cdots b_1 b_0)_2 \le \sum_{i=1}^{s} (b_{i} \cdots b_1 b_0)_2
  \le \sum_{i \ge 1} \left\lfloor \frac{M}{2^i} \right\rfloor \le N.
\]
Thus the running time is still $O(N)$ and the space is still $O(\log^2 N)$.
However, the totality of the rotations can not be done with a constant number of rounds of swaps,
and the extension of this idea to $k > 2$ is non-trivial.
So we seek some other method to handle the case where $N$ is not a power of $k$.

\section{A number-theoretic approach for the $k$-way perfect shuffle, valid for any $N$}

In this section we make no restriction on $N$, other than it be a multiple of $k$.
We will show how to use some elementary number theory to obtain explicit expressions
  for the two involutions guaranteed by Corollary \ref{cor:involution}.

Let $g = \gcd(x,m)$.
Define $J_r : \mathbb{Z}_m \rightarrow \mathbb{Z}_m$ by
\[
J_r(x) = g \left( r \left( \frac{x}{g} \right)^{-1} \bmod{\frac{m}{g}} \right).
\]

The following lemma has two important consequences:
\begin{itemize}
\item
If $\gcd(r,m) = 1$, then $J_r$ is an involution.
\item
If $m = kn-1$, then $\gcd(k,m) = 1$ and thus
  $J_k(J_1(x)) = kx$ for all $x \in \mathbb{Z}_m$.
That is, $J_k(J_1(x))$ computes the $k$-way perfect shuffle.
\end{itemize}

\begin{lemma}
If $\gcd(r,m) = \gcd(s,m) = 1$, then
\[
J_r(J_s(x)) = x r \left( s^{-1} \bmod{\frac{m}{g}} \right) \bmod{m}.
\]
\end{lemma}

\begin{proof}
Let $g = \gcd(x,m)$.
First note that, since $\gcd \left( \frac{m}{g}, \left( \frac{x}{g} \right) \bmod \frac{m}{g} \right) = 1$,
we also have $\gcd \left( \frac{m}{g}, \left( \frac{x}{g} \right)^{-1} \bmod \frac{m}{g} \right) = 1$.
Thus
\begin{eqnarray*}
\gcd(m,J_s(x))
& = & \gcd \left( m, g \left( s \left( \frac{x}{g} \right)^{-1} \bmod{\frac{m}{g}} \right) \right) \\
& = & g \cdot \gcd \left( \frac{m}{g}, \left( s \left( \frac{x}{g} \right)^{-1} \bmod{\frac{m}{g}} \right) \right) \\
& = & g \cdot \gcd \left( \frac{m}{g}, s \left( \left( \frac{x}{g} \right)^{-1} \bmod{\frac{m}{g}} \right) \right) \\
& = & g \cdot \gcd \left( \frac{m}{g}, \left( \frac{x}{g} \right)^{-1} \bmod{\frac{m}{g}} \right) \\
& = & g \cdot \gcd \left( \frac{m}{g}, \left( \frac{x}{g} \right) \bmod{\frac{m}{g}} \right) \\
& = & g \cdot \gcd \left( \frac{m}{g}, \frac{x}{g} \right) \ \ = \ g. \\
\end{eqnarray*}
We can now compute the composition:
\begin{eqnarray*}
J_r(J_s(x))
& = & g \left( r \left( \frac{J_s(x)}{g} \right)^{-1} \bmod{\frac{m}{g}} \right) \\
& = & g \left( r \left( \frac{g \left( (s(x/g)^{-1} \bmod{(m/g)} \right) }{g} \right)^{-1} \bmod{\frac{m}{g}} \right) \\
& = & g \left( r \left( s^{-1} \frac{x}{g} \bmod{\frac{m}{g}} \right) \bmod{\frac{m}{g}} \right) \\
& = & g \left( \frac{rx}{g} \left( s^{-1} \bmod{\frac{m}{g}} \right) \bmod{\frac{m}{g}} \right) \\
& = & rx \left( s^{-1} \bmod{\frac{m}{g}} \right) \bmod{m}.
\end{eqnarray*}
\end{proof}

The pseudo-code below outlines our algorithm and examples of its outputs are illustrated
in Figure \ref{fig:examplesA}. \\

\begin{tabular}{lcr}
\begin{minipage}[t]{2in}
\begin{tabbing}
xxx \= xxx \= xxx \kill
\textsc{ShuffleA}($M, k$) \\
\> \textsc{modInv}$(M,1)$ \\
\> \textsc{modInv}$(M,k)$ \\
\end{tabbing}
\end{minipage}
& \hspace{0.3in} &
\begin{minipage}[t]{3in}
\begin{tabbing}
\textsc{modInv}$(M,r)$ \\
xxx \= xxx \= xxx \kill
\> \textbf{for} $i \in \{1,2,\ldots,k-2\}$ \textbf{do} \\
\> \> $g := \gcd(x,kM-1)$ \\
\> \> $j := (r(x/g)^{-1}) \mod ((kM-1)/g))$ \\
\> \> \textbf{if} $i<j$ \textbf{then} $A[i] :=: A[j]$ \\
\end{tabbing}
\end{minipage}
\end{tabular}

\begin{figure}[t]
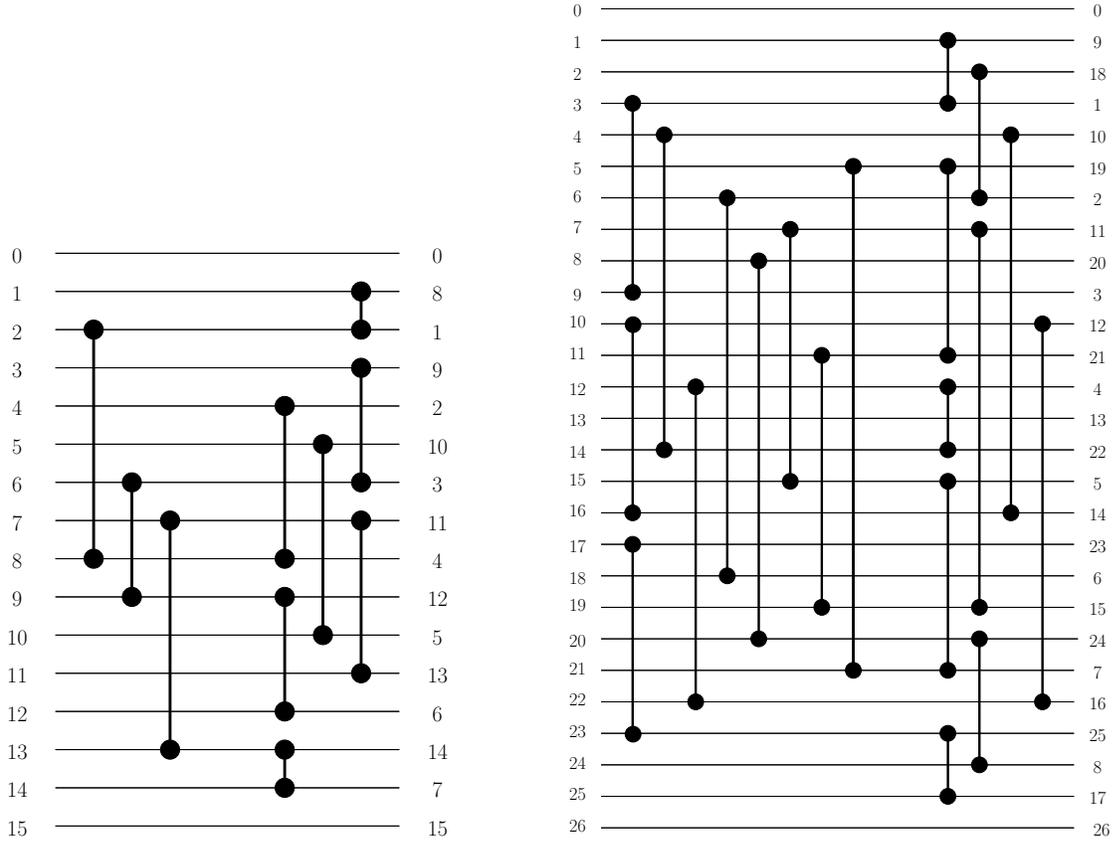

\begin{center}
\scalebox{0.40}{\input{YE-2-4A.pstex_t}}
\hspace{0.5in}
\scalebox{0.33}{\input{YE-3-3A.pstex_t}}
\end{center}
\caption{A network corresponding to the number theoretic algorithm for (a)
  $k = 2$, $N = 16$ on the left, and (b) $k=3$, $N = 27$ on the right.}
\label{fig:examplesA}
\end{figure}

The involutions $J_1$ and $J_k$ take longer to compute than the
corresponding bit-reversal involutions discussed in the previous section.
In particular, a modular inverse can take time $O(\log N)$ time to compute,
via the extended Euclidean algorithm, giving a total worst case running time of $O(N \log N)$.

However, this will be overly pessimistic in practice since not all $\gcd$ calculations
will be worst case.  If we count the number of arithmetic operations that are actually used
by the algorithm, then the true running time appears \emph{experimentally} to be $O(N \log \log N)$ with
the worst case examples occurring when $N$ is the product of a small number of distinct odd primes.

\section{Final Remarks and Acknowledgement}

Here we mention some open problems.
It would be of interest to find other natural classes of permutations for which the
involutions of Corollary \ref{cor:involution} can be efficiently computed.
Note that the number of swaps used by the first algorithm when
  $N = 27$ and $k = 3$ is 18, but the number used by the second algorithm
  is 20 (See (b) in Figures \ref{fig:examples} and \ref{fig:examplesA}).
For a general $N$ and $k$, what is the least number required?
Finally, what is the expected running time of \textsc{ShuffleA}?

We thank Dominique Gouyou-Beauchamps
  for pointing out an erroneous sentence in an earlier version of this paper.
We also thank the referees for their helpful comments.

\bibliographystyle{plain}
\bibliography{YER_Ellis}

\end{document}